\documentclass[aps,pra,twocolumn,showkeys]{revtex4-1}
\usepackage{amsmath,paralist,amsthm,comment,amssymb}
\usepackage{times,color}
\usepackage{epsfig}

\DeclareMathOperator{\tr}{tr}

\DeclareMathOperator{\supp}{supp}
\DeclareMathOperator{\rk}{rk}

\newcommand{\mc}[1]{\mathcal{#1}}

\newcommand{\spn}[1]{\left\langle{#1}\right\rangle}
\newcommand{\C}{\mathbb{C}}
\newcommand{\F}{\mathbb{F}}
\newcommand{\nix}[1]{}

\newcommand{\ket}[1]{|#1\rangle}

\newtheorem{theorem}{Theorem}
\newtheorem{corollary}[theorem]{Corollary}
\newtheorem{lemma}[theorem]{Lemma}
\newtheorem{prop}{Proposition}
\newtheorem{defn}{Definition}

\newtheorem{remark}{Remark}

\newcommand{\ba}{\begin{array}}
\newcommand{\ea}{\end{array}}
\newcommand{\ben}{\begin{eqnarray}}
\newcommand{\een}{\end{eqnarray}}

\newcommand{\be}{\begin{eqnarray*}}
\newcommand{\ee}{\end{eqnarray*}}

\begin{document}
\title{ Quantum Codes and Symplectic Matroids }

\author{Pradeep Sarvepalli}  \email[]{pradeep@phas.ubc.ca}
\affiliation{Department of Physics and Astronomy, University of British Columbia, Vancouver, BC V6T 1Z1}

\date{April 6, 2011}
\begin{abstract}
The correspondence between linear codes and representable matroids is well known. But a 
similar correspondence between quantum codes and matroids is not known. We show that 
representable symplectic matroids over a finite field $\mathbb{F}_q$ correspond to 
$\mathbb{F}_q$-linear quantum codes. Although this connection is straightforward, it does not 
appear  to have been made earlier in literature.  The correspondence is made through 
isotropic subspaces. We also show that the popular Calderbank-Shor-Steane (CSS) codes are 
essentially the homogenous symplectic matroids while the graph states, which figure so 
prominently in measurement based quantum computation,  correspond to  a special class of 
symplectic matroids, namely Lagrangian matroids. This  association is useful in that it 
enables the study of symplectic matroids in terms of quantum codes and vice versa. 
Furthermore, it has application in the study of quantum secret sharing schemes. 
\end{abstract}

\pacs{}
\keywords{quantum codes, symplectic matroids, Lagrangian matroids, graph states, quantum secret sharing, quantum 
cryptography}

\maketitle
\section{Introduction}
Matroids are mathematical structures that abstract the idea of independence. Originally,
introduced by Whitney, they have since found applications in various fields most notably
in algorithms, combinatorial optimization, graphs, cryptography, coding theory to name a
few. A particular class of matroids called the representable matroids are closely related
to error-correcting codes. In fact, the so-called representations of these matroids give rise to
linear codes; further, one can obtain matroids from linear codes. This correspondence goes 
much deeper in that certain invariants of the code are essentially invariants of the matroid
as well. (Most well-known is the connection between 
the weight enumerator of a linear code and the Tutte polynomial of the matroid associated
to the code.) 

Given these associations one is tempted to ask if we can  find a similar 
correspondence between quantum codes and (a class of) matroids? The answer to this
question, as we shall see, is surprisingly simple and straightforward. In fact, it goes back
to the many ways we can view matroids. But 
this connection does not appear to have  made  in the literature so far.

The main results of this paper are the correspondence between quantum codes and matroids,
and applications of this correspondence.
Strictly speaking we establish a correspondence between quantum codes and objects which 
are more general than matroids, called the symplectic matroids. Symplectic matroids generalize
matroids,  although their definition is somewhat more complicated than
matroids. Our result has important applications. It can be used to study quantum codes 
using matroids and vice versa. 
We
also find an application for these results in quantum secret sharing. We show how certain 
symplectic matroids induce quantum secret sharing schemes. There are many important open 
problems that arise with this connection and we are hopeful that further research along these
lines will be fruitful for either communities of quantum information theorists and matroid
theorists. 

\section{Background}

\subsection{Symplectic matroids}
Our presentation of the symplectic matroids follows the exposition in
\cite{borovik03} very closely. Consider the sets $[ n ]=\{1,\ldots, n \}$ and 
$[n]^\ast =\{1^\ast,\ldots, n^\ast \}$. Let $J= [n]\cup [n]^\ast$ and define an involution on
$J$ as 
\ben
\ast:J \rightarrow J, \mbox{ where } i\mapsto i^\ast \mbox{ and } (i^\ast)^\ast=i 
\label{eq:inv}
\een
This map can be extended naturally to subsets of $J$. A set $S\subset J$ is said to be admissible if 
$S\cap S^\ast =\emptyset$. A transversal is an admissible set of size $n$; it is a maximal 
admissible set. 
 Consider now the group of permutations on the set $J$;
a permutation is said to be admissible if it commutes with the involution. This group of admissible permutations  on $J$, denoted as $W$, is the
hyperoctahedral group of symmetries, the group of symmetries of the hypercube $[-1,1]^n$
in $n$-dimensions. 

Consider the  ordering of the elements of $J$ as given by
\ben
n>n-1>  \cdots > 2 >1 > 1^\ast >2^\ast \cdots > n^\ast.
\een
We now define another ordering on the set $J$ by means of the admissible permutation 
$w \in  W$. We say that $i\leq^w j$ if and only if $w^{-1}i \leq w^{-1} j $.
Let $w$ be given by the following permutation:
\be
\left(\ba{cccccccc} 1& 2 &\ldots & n & n^\ast & \ldots & 2^\ast & 1^\ast \\
i_1& i_2 &\ldots & i_n& i_{n+1}&\ldots & i_{2n-1}&i_{2n}\ea \right)
\ee
This permutation induces the ordering $\prec$ given by 
\be
i_1 < i_2 < \cdots< i_n<i_{n+1}<\cdots <i_{2n}.
\ee
Clearly, $\prec$ induces an 
ordering on the subsets of $J$. It can also be used to order subsets $A,B\subset J$.
Given two subsets $A=\{a_1,\ldots, a_m \}$, and 
$B =\{b_1,\ldots, b_m \}$, we say that $A\leq^w B$ if and only if $ a_i\leq^w b_i$,
where we assumed that $A$ and $B$ have been ordered as $\{a_1\prec a_2\prec \cdots \prec a_m \}$ and $\{b_1\prec b_2\prec \cdots \prec b_m \}$ respectively.

\begin{defn}[Symplectic matroids]
Let $J_k$ be the collection of admissible $k$-subsets of $J$ and 
$\mc{B}\subseteq J_k$.  A tuple $(J, \ast, \mc{B})$ is  a symplectic matroid if and only if it satisfies the following
condition:

For every admissible ordering of the set $J$, there exists a unique maximal set $B\in \mc{B}$
such that for all $A \in \mc{B}$, we have $A\prec B$.
\end{defn}

The condition mentioned above is often called the {\bf \textsl{Maximality condition}}.
The elements of $\mc{B}$ are called bases while $\mc{B}$ itself is called collection of the bases  of
the symplectic matroid. The cardinality of the bases is called is the rank of the matroid. 
(All the bases have the same size.) If the rank of the symplectic matroid is
the maximal value of  $n$, then it is said to be a Lagrangian matroid. 

\begin{remark}
Suppose we  set $J=[n]$ and  instead of $W$, we consider the symmetric group of all
permutations, ie. all permutations on $J$ are admissible. 
 Then the tuple $(J,\mc{B})$, where $\mc{B}$ is a collection of $k$-subsets of $J$, is a matroid if and only if $\mc{B}$ satisfies  the Maximality condition. 
In this case the involution plays no role. 
It is common in this case to refer to $J$ as the ground set.
\end{remark}

\subsection{Representable symplectic matroids}
It is often convenient to 
deal with what are known are as the representations of a matroid. These representations 
provide us with  
a concrete object to work with and study the properties of the matroid. An ordinary matroid is said
to have a representation if the elements of the ground set can be identified with the columns
of a matrix (typically over some field) such that columns indexed by the bases are 
maximally linearly independent columns of that matrix. 

Some symplectic matroids can also be endowed with representations. In this case instead of a standard vector space (with an orthogonal basis), we consider a symplectic vector space. That is a space of dimension 
$2n$ and endowed with a symplectic form $\langle \cdot ,\cdot \rangle$, whose basis $\{e_1,\ldots, e_n, e_1^\ast,\ldots, e_n^\ast \}$ satisfies the following relations:
\ben
\langle e_i,e_j \rangle&=&0, i\neq j^\ast\\
\langle e_i,e_i^\ast\rangle&=& -\langle e_i^\ast,e_i\rangle= 1
\een

\begin{defn}
A vector space $V$ over a field $\F$ is said to be isotropic if and only if for any
$u,v\in V$ we have $\langle u , v\rangle = 0 $, where $\spn{\cdot ,\cdot}$ is the
inner product.
\end{defn}

Let $U$ be an isotropic subspace of a symplectic vector space.  Suppose we write down a basis of this 
isotropic space as the rows of a  matrix    $M=[A|B ] \in \F^{k\times 2n}$, where $k$ is the dimension of $V$;  then we must have  $AB^t=BA^t$. 
Index the  columns of $M$ by the set 
$J=[n]\cup[n]^\ast$.
Let $B\subset J$ such that $B\cap B^\ast=\emptyset$ and $|B|=k$. Then if the $k\times k $ 
minor of $M$ indexed by $B$ is nonzero, then we say that $B$ is a basis of $M$. Let $\mc{B}$ denote the collection of bases of $M$. 
Then $(J,*,\mc{B})$  
is a symplectic matroid over $\F$.
\begin{prop}[\cite{borovik03}]\label{prop:repSM}
Let the row space of $M=[A|B] \in \F^{s\times 2n}$ be an isotropic subspace with respect to
a symplectic 
form. Then $M$ is the representation of a symplectic matroid.
\end{prop}
A symplectic matroid is said to be homogenous if for every basis $B\in \mc{B}$,
we have $|B\cap [n]|$ is same. For such a matroid $|B\cap [n]^\ast|$ is also independent of 
$B$. If such a matroid is representable then its
representation is of the form 
\be
M =\left[\ba{cc} X&0 \\0 & Z \ea \right],
\ee
where $XZ^t=0$.
For the rest of the discussion in this paper we will assume that the matroid 
representations
are over a finite field $\F_q$; occasionally we specialize to the case of $\F_2$ for
simplicity.

\section{Connections with quantum codes}
We recall some of the notions relevant for quantum codes. 
We will confine our discussion to additive quantum codes, in particular to stabilizer codes.  
Interested readers can find more details in \cite{calderbank98,gottesman97} for binary quantum codes and \cite{rains99, ashikhmin01,grassl03,ketkar06} 
 for nonbinary versions. 
 Let $q$ be the power of a prime $p$
 and  $\F_q$ a finite field. Suppose that 
$\C^{q}$ denotes the
$q$-dimensional complex vector space. Fix a basis for $\C^q$ as $B=\{\ket{x}\mid x\in \F_q \}$. We define error operators on $\C^q$ as $X(a)\ket{x} =\ket{x+a}$ and $Z(b)\ket{x}=\omega^{\tr_{q/p}(bx)}\ket{x}$. Error operators on $n$ such $q$-level quantum systems are operators on  $\C^{q^n}$ and are obtained as 
tensor products of the operators on $\C^q$. These error operators form the generalized Pauli group which is denoted as
\ben
\mc{P}_n =\{ \omega^cX(a_1)Z(b_1)\otimes \cdots \otimes X(a_n)Z(b_n)\},\label{eq:Pauligrp}
\een
where $\omega=e^{j2\pi/p}$.

 An $((n,K,d))_q$ quantum code is a $K$-dimensional 
subspace of the $q^n$-dimensional complex vector space $\C^{q^n}$ and 
able to detect all errors on fewer than $d$ subsystems. When $K=q^k$, it is also denoted as
an $[[n,k,d]]_q$ code. A stabilizer code
is the joint eigenspace of an abelian subgroup of $\mc{P}_n$. The subgroup is 
called the stabilizer of the code. For a nontrivial quantum code, the stabilizer does 
not have any scalar multiple of identity other than the identity itself.

By defining a map 
between the Pauli group and the
vector spaces over $\F_q^{2n}$, we can establish a correspondence between quantum codes and
classical codes. This correspondence with the classical codes
has been used extensively in the study of quantum codes \cite{calderbank98,gottesman97,rains99, ashikhmin01,grassl03,ketkar06}.
 An element 
$\omega^cX(a_1)Z(b_1)\otimes \cdots \otimes X(a_n)Z(b_n)$ in $\mc{P}_n$ is mapped to $(a_1,\ldots,a_n|b_1,\ldots, b_n) \in \F_q^{2n}$. Under this
mapping the stabilizer of the quantum code is mapped to a $\F_p$-linear subspace of 
$\F_q^{2n}$. If the image of the stabilizer is also an $\F_q$-linear subspace then we say that
it is an $\F_q$-linear quantum code. 
In this paper we restrict our attention to $\F_q$-linear codes only. 
The image of a set of generators of the stabilizer under this map is often called a
stabilizer matrix.

The relevant bilinear form that we endow $\F_q^{2n} $ with is the symplectic inner product
defined as follows. Let $u,v$  be two vectors in $\F_q^{2n}$ where
$u=(a|b)=(a_1,\ldots, a_n|b_1,\ldots,b_n)$ and $v=(c|d)=(c_1,\ldots,c_n|d_1,\ldots,d_n)$.
 Then their symplectic inner product is defined as 
\ben
\langle u | v \rangle_s  =   ( a\cdot d -c\cdot b ).
\een

It is $\F_q$-linear in the sense that 
$\langle u | v \rangle_s  =0$ if and only if $\langle \alpha u | \beta v \rangle_s  =0 $
for all $\alpha,\beta \in \F_q$.
It can be easily checked that this form is 
asymmetric as 
$\langle u | v \rangle_s = - \langle v | u \rangle_s$. Denoting the standard basis of 
$\F_q^{2n}$ as $\{e_{i},\ldots,e_n, e_1^{\ast},\ldots, e_n^\ast\}$, we can check that 
$\langle e_i|e_j\rangle_s = 0 $ for $i\neq j^\ast$,  and
$\langle e_i|e_i^\ast\rangle_s = 1 $. 

In this case the stabilizer matrix of an  $\F_q$-linear $[[n,k,d]]_q$ quantum
code defines an isotropic subspace of $\F_q^{2n}$ and is an element of $\F_q^{(n-k)\times 2n}$. This gives us the following result:

\begin{prop}[\cite{calderbank98,gottesman97}]\label{prop:isotropic}
Let $Q$ be an $[[n,k,d]]_q$ $\F_q$-linear quantum code, then the row space of the stabilizer matrix of the
code defines an isotropic subspace of dimension $n-k$.
\end{prop}
Putting together with our discussion on the representations of symplectic matroids the 
following result  is immediate.

\begin{theorem}\label{th:qeccMat}
Let $Q$ be an $[[n,k,d]]_q$ $\F_q$-linear quantum code. Then $Q$ induces a representable symplectic matroid over $\F_q$ of rank $n-k$. If $Q$ is a CSS code it induces a representable 
homogenous matroid.
\end{theorem}
\begin{proof}
This is an immediate  consequence of Proposition~\ref{prop:isotropic} and Proposition~\ref{prop:repSM}. The stabilizer matrix of a CSS code is precisely the same form as in equation~\eqref{eq:repHSM},  (see \cite{calderbank98}) and consequently, it induces a homogeneous symplectic matroid.
\end{proof}
It turns out the distance of the quantum code is related to the cardinality of the 
circuit of smallest size but to prove it more precisely we must wait till we have 
a few more results in hand.

With appropriate permutation of the
columns of its representation a representable Lagrangian matroid can  be put in the form $\left[\ba{cc}I&A\ea\right]$,
where $A$ is a symmetric matrix. If $A$ is such that its diagonal is all zero then we
can identify it with adjacency matrix of a (weighted) graph. 
Recall that a graph state over $\F_2$ is defined as the quantum
state  whose stabilizer is given by 
\ben
S= \spn{K_v \mid v\in V(G); K_v= X_v\!\prod_{u\in N(v)}\! Z_u}\label{eq:gstateF2}
\een
where  $V(G)$ is the vertex set of $G$ and $N(v)$ is the set of neighbors of $v$.
If $G$ is a weighted graph we can define a graph state over $\F_q$ with stabilizer  as 
follows:
\ben
S= \spn{K_v \mid v\in V(G) ; K_v= X_v(1)\!\prod_{u\in N(v)}\! Z_u(w_{uv})}
\label{eq:gstateFq}
\een
where  
$w_{uv}$ is the weight of the edge $uv$.
See  \cite{schlingemann00,bahramgiri06} for more details on 
nonbinary graph states.

Since a stabilizer state corresponds to an $[[n,0,d]]_q$ code, Theorem~\ref{th:qeccMat}
implies the following:
\begin{corollary}
Every graph state  induces a representable Lagrangian matroid. 
\end{corollary}

We pause to note a few differences with respect to the correspondence between matroids and
classical codes. In case of classical codes the independent sets correspond to a subset of
errors that are detectable. The codewords correspond to dependent sets. Further, the minimally dependent codewords characterize the matroid completely. (A minimal codeword $x$
does not contain the support of any other codeword $y$, unless $y$ is the scalar of $x$.)
The supports of these minimal  codewords are called circuits of the associated matroid. 
The concept of circuits can be generalized  for symplectic matroids
but circuits are most useful in the characterization of special cases of  symplectic matroids
 such as Lagrangian matroids. 

Classical (linear) codes have well-defined dual codes, on the other hand, there is no 
equivalent notion of a dual quantum code for a quantum code be it linear or additive. And not surprisingly,
we find that a similar notion of duality is lacking for symplectic matroids. There has been
a suggestion by Borovik \cite{chow03} to use the involution defined in equation~\eqref{eq:inv} for defining duals, however this
suggestion seems to be most fruitful for the Lagrangian matroids and not for the 
general symplectic matroids. 

\begin{remark}[Quantum codes and ordinary matroids]
Suppose that an $[[n,k,d]]_q$ quantum code is $\F_{q^2}$-linear, then we can also associate 
an ordinary matroid to that code in addition to a symplectic matroid. In this case the 
stabilizer matrix can be represented by a $(n-k)/2\times n $ matrix over $\F_{q^2}$.
In this particular instance, we can associate 
the vector matroid of this matrix to the quantum code. Thus $\F_{q^2}$-linear codes afford
multiple associations to matroids. 
\end{remark}

\subsection{New quantum codes from graphical symplectic matroids}

Quantum codes from graphs have been studied extensively in the context of fault tolerance. 
We now propose a new class of quantum codes induced by graphs by way of symplectic matroids. 
These are derived from the graphical symplectic matroids proposed by Chow \cite{chow03}.

The graphical symplectic matroids are defined as follows. Let $G$ be a graph of $n$ edges. Label the edges of the graph by a transversal $T\subset [n]\cup [n]^\ast $. (Recall that a
transversal in an admissible set of size $n$.)
A cycle in $G$ is called balanced if there are an even number of edges labeled with elements
from $[n]^\ast$, otherwise it is said to be unbalanced.
An admissible set $S\subset [n]\cup [n]^\ast$ is an independent
set if it is either a forest or every connected component is a tree plus an edge such that the cycle has an odd number of edges in $[n]^\ast$. It is the import of  
\cite[Theorem~2]{chow03}, that the maximal independent sets form the bases of a
 symplectic matroid.

 Assuming a connected graph, we can state some properties of
these symplectic matroids. If the graph is a tree, then the rank of the symplectic matroid
is $|V|-1$. If the graph is not a tree, then the rank is $|V|$. If these matroids are
representable then we have a quantum code from Theorem~\ref{th:qeccMat}. However, all 
graphic symplectic matroids are not representable \cite{chow03}. Supposing that it is 
representable then the code has parameters $[[|E(G)|, |E(G)|-|V(G)|, d]]_q$, where $d\geq $
 the smallest cycle in the graph.

 As an example, the complete graph on three vertices 
is identical to the graph state on that graph.  For dense graphs the associated  codes are not likely 
to have good distance. On the other hand, sparse graphs might lead to good quantum
codes.  The main 
reason for proposing these codes is to illustrate the 
possibility that matroids can provide new perspectives on quantum codes.

\subsection{New symplectic matroids via quantum codes}
Unlike matroids, symplectic matroids are a little more restricted in obtaining new symplectic
matroids from existing ones. There are a however, few constructions known for constructing symplectic matroids: 
contraction,  truncation, Higgs lift and direct sum \cite{borovik03}. 
For the
representable symplectic matroids which correspond to $\F_q$-linear quantum codes
one can relate these constructions to familiar coding theoretic operations.

Consider a symplectic matroid of rank $k$ whose collection of bases are given by $\mc{B}$. 
Contraction (along) $a\in J$ is defined by the following operation:
\ben
\mc{B}'= \{ B \mid (B\cup \{a\}) \in \mc{B}\},\label{eq:contraction}
\een
where $\mc{B}'$  is the  collection of bases of the resulting symplectic matroid.
This translates to obtaining an $[[n-1,k]]_q$ from an $[[n,k]]_q$ code.
Truncation modifies $\mc{B}$  as 
\ben
\mc{B}'= \{ A \in J_{k-1} \mid A \subset B \in \mc{B}\}.\label{eq:truncation}
\een
In coding theoretic terms this is equivalent to obtaining an $[[n,n-k+1]]_q$ quantum code
from an $[[n,n-k]]_q$ quantum code.

On the other hand 
deletion corresponds to puncturing on the underlying code and as this does not always preserve
a self-orthogonality of the code, this construction does not generalize.
An interesting method for constructing new symplectic matroids is the so-called  Higgs lift \cite{borovik03}. 
This corresponds to obtaining an $[[n,k-1]]_q$ code from an $[[n,k]]_q$ code.

Two symplectic matroids can be combined to give rise to a third matroid in many ways. The
simplest method is the direct sum method. 
Concatenation is a popular method to construct new codes and if done appropriately it gives
rise to another self-orthogonal code. There are many flavors of concatenating quantum
codes \cite{rains99,grassl05}.  These constructions can be translated to equivalent constructions of 
symplectic matroids.

\subsection{Transformations of symplectic matroids}
One of the most studied equivalence of quantum codes is local equivalence, especially
local Clifford equivalence. It is natural to ask if this corresponds to any equivalence
on the associated symplectic matroids. The (representable) symplectic matroids are not going 
to be preserved under local Clifford operations in general. 
This can be checked with the complete graph on 3 vertices and the graph obtained
by local complementation at any of the vertices.  The symplectic matroid associated with 
the line graph on 3 vertices has the representation 
\be
\left[\ba{ccc|ccc}1&0&0&0&1&1 \\
0&1&0&1&0&0\\
0&0&1&1&0&0\ea\right]
\ee
with the associated bases being $\{ \{1,2,3 \}, \{1^\ast, 2^\ast, 3 \}, \{1^\ast, 2, 3^\ast \}  \}$
On the other hand,  the symplectic matroid of graph state on the complete graph on three vertices which is 
local Clifford equivalent to it has the representation
\be
\left[\ba{ccc|ccc}1&0&0&0&1&1 \\
0&1&0&1&0&1\\
0&0&1&1&1&0\ea\right]
\ee
This symplectic matroid has its collection of bases 
$\{ \{1,2,3 \}, \{1^\ast, 2^\ast, 3 \}, \{1^\ast, 2, 3^\ast \}, \{1,2^\ast,3^\ast \} \}$.
This prompts the question is there an operation by which we can express this transformation
of the symplectic matroid in terms of an operation on its bases?

One of the methods to obtain an equivalent symplectic matroid is via the torus action 
defined as follows. Let  
$[A|B]$ be the representation of a symplectic matroid.
 Then for any invertible $n\times n $ diagonal matrix $T$, the representation 
 $[A T ^{-1} |B T ]$  is also a representation of the symplectic matroid. 
The torus action gives rise to an equivalent quantum code with the same parameters.
Furthermore, the weight distribution of the code is unchanged under the torus action. 

\subsection{Representable homogeneous symplectic matroids}

Given a symplectic matroid define a circuit to be 
 a minimally dependent admissible subset of $J$. 
Then we have the
following characterization for the homogenous symplectic matroids. 
These results will be needed later in the section on quantum secret sharing. 

\begin{lemma}\label{lm:repCkts}
Every circuit of a representable homogeneous symplectic matroid consists of either elements in
$[n]$ or $[n]^\ast$. 
\end{lemma}
\begin{proof}
Suppose that there is a minimally dependent admissible set $C\subset J$ such that 
$C\cap [n] \neq \emptyset$ and $C\cap [n]^\ast \neq \emptyset$. Without loss of 
generality assume that $C=\{1,\ldots,m,(m+1)^\ast,\ldots,p^\ast \}$. Assume that the
representation of the symplectic matroid is given by 
\ben
M= \left[ \ba{c|c} X& 0 \\0 & Z \ea\right].\label{eq:repHSM}
\een
As $C$ is a circuit, there exists a linear combination of the columns $\{1,\ldots, m \}$
and the columns $\{(m+1)^\ast,\ldots, p^\ast \}$. However given the fact that the
representation of the matroid is of the form equation~\eqref{eq:repHSM}, the columns 
$\{1,\ldots, m \}$ and $\{(m+1)^\ast,\ldots, p^\ast \}$ are linearly dependent as well.
But this implies that $C$ is not a minimally dependent set. Therefore every circuit of 
the homogenous symplectic matroid is either a subset of $[n]$ or $[n]^\ast$ but not both.
\end{proof}

\begin{theorem}
Representable homogenous symplectic matroids, satisfy the Circuit elimination property: 
If $C_1$, $C_2 \in \mc{C}$, such that $e\in C_1\cap C_2 $ and $C_1\cup C_2$
is admissible, then there exists a circuit $C\in \mc{C}$ such that $C \subseteq (C_1\cup C_2)\setminus \{ e\}$. 
\end{theorem}
\begin{proof}
Let $C_1$ and $C_2$ be two circuits of $\mc{M}$. By Lemma~\ref{lm:repCkts}, every such 
circuit consists of elements in $[n]$ or $[n]^\ast$. Suppose that $C_1\cap C_2\neq \emptyset$.
Then this is possible if and only if both $C_1,C_2\subset [n]$ or $C_1,C_2\subset[n]^\ast$.
Without loss of generality assume that $C_1,C_2\subset [n]$. Let $e\in C_1\cap C_2$.
Then $e$  can be expressed a linear combination of columns in $C_1\setminus \{e \}$
as well as $C_2\setminus \{e\}$. It is then immediate that $C_1 \cup _2\setminus \{e \}$
is a dependent set and must contain a minimal dependent set equivalently a circuit in 
$[n]$, which is clearly an admissible set. 
Thus representable homogenous symplectic matroids satisfy the circuit elimination property.
\end{proof}

Before we move to some applications of these results, we raise the question we address
the issue of  invariants for the symplectic matroids. 

\subsection{Invariants for symplectic matroids}
An important invariant associated with matroids is the rank polynomial. As a weight 
enumerator captures many of the invariants of the code (such as distance), the rank 
polynomial encodes information about many invariants of the matroids. The rank polynomial has 
been related to other polynomials of interest such as Tutte polynomial of a graph, the 
Kauffman polynomial of a knot, the partition function  and has been studied extensively in 
view of its relevance to complexity theory. But from a coding theoretic point of view the
weight enumerator and the rank polynomial are closely related. 
All this brings up the question if there are similar 
polynomials for the symplectic matroids which are of interest to quantum codes. 
A general answer to this question eludes us, but when we focus our attention to the 
Lagrangian matroids, we can partially answer this question.

In \cite{bouchet05}, Bouchet studied graph polynomials for isotropic systems that are related to the Tutte polynomial of an associated graph. 
 Isotropic systems are essentially 
Lagrangian matroids. Consequently the following Tutte-Martin polynomials as defined by Bouchet are
only defined for  Lagrangian matroids. 

\begin{defn}[Restricted Tutte-Martin polynomial]
Let $L$ be a Lagrangian matroid. Define the restricted Tutte-Martin polynomial as 
\ben
m(L; x) &=& \sum_{S\in J_n} (x-1)^{n-\rk(S)}.\label{eq:rtmPoly}
\een
where $n=\rk(L)$.
\end{defn}

We could attempt to define a similar polynomial for symplectic matroids that are not Lagrangian. For a symplectic matroid,  $L$ we define the restricted Tutte-Martin polynomial as
\ben
m(L; x) &=& \sum_{S\in J_k} (x-1)^{k-\rk(S)}.\label{eq:rtmPoly}
\een
where $k=\rk(L)$.

Suppose $M$ is a representable  Lagrangian matroid, with representation $[I | A]$, for some symmetric matrix, $A$. 
Then its restricted Tutte-Martin polynomial 
is the same as the interlace polynomial of a graph $G$ with adjacency matrix $A$.
Note that the interlace polynomial $q_N(x)$ is defined as \cite{aigner04}
\ben
q_N(G; x) &= &\sum_{S \subseteq V(G)} (x-1)^{\text{corank}(G(S))},
\een
where $G(S)$ is the subgraph of $G$ induced by $S$.
Bouchet who originally defined the restricted Tutte-Martin polynomial gave it in a slightly
different form.
\nix{
The restricted Tutte-Martin polynomial of a representable Lagrangian matroid $L$ is defined as follows:
\ben
M(L,u;x) &=& \sum_{c\in \F^n : c_i\neq u_i} (x-1)^{\dim(V \cap \hat{c} )},
\een
where $\F=\F_{q^2}^\times$ and $\hat{c}= \{ y \in \F_{q^2}^n \mid y_i = 0 \text{ or } y_i =c_i  \}$ and $V$ is the isotropic vector space that represents $L$.

\begin{defn}[Global Tutte-Martin polynomial\cite{bouchet05}]
The global Tutte-Martin polynomial of a representable Lagrangian matroid $L$ is defined as follows:
\ben
M(L;x) &=& \sum_{c\in \F^n} (x-2)^{\dim(V \cap \hat{c} )},
\een
where $\F=\F_{q^2}^\times$ and $\hat{c}= \{ y \in \F_{q^2}^n \mid y_i = 0 \text{ or } y_i =c_i  \}$.
\end{defn}

}

Recent work \cite{danielsen10} has made the connection between interlace polynomial and orbits
of quantum states and codes under edge local complementation.
Perhaps the most famous polynomial associated to matroids is the rank polynomial or 
the Tutte polynomial. It does not seem possible to define a Tutte polynomial for a 
symplectic matroid in general  and 
might require an expansion of the definition of symplectic matroid.

\section{Application for quantum secret sharing}
In \cite{pre10}, connections between matroids and quantum secret sharing schemes were 
investigated. It was shown that identically self-dual matroids induce quantum 
secret sharing schemes thereby this establishing a connection between matroids and
quantum secret sharing schemes. However, it was somewhat limited in that 
only quantum secret sharing schemes that are realized using a CSS code were within that
correspondence. In present section we intend to make this matroidal correspondence stronger 
by including a larger class of schemes some of which can be realized by non-CSS codes.

Given a  Lagrangian matroid ${L}$ whose collection of bases is $\mc{B}$, we can 
define the dual matroid as follows. 
The collection of bases of the dual matroid are given by $B^\ast = \{B^\ast \mid B\in \mc{B} \}$. Similarly, the collection of circuits of the dual matroid are given by 
$\mc{C}^\ast =\{C^\ast \mid C\in \mc{C} \}$. Elements of $\mc{C}^\ast$ are also called cocircuits 
of $\mc{L}$.

Let ${L}$ be a self-dual Lagrangian matroid, then we define an access structure from the
circuits of $\mc{L}$  as follows. Define the map $\varphi:[n]\cup [n]^\ast\rightarrow [n]$
where 
\ben
\varphi(i)=\left\{ \ba{ll}i & \text{ if } i \in [n]\\
i^\ast & \text{ if } i \in [n]^\ast\ea\right.
\een
We obtain an access structure by considering $i\in [n]$
as the dealer. The induced minimal  access structure is given as
\ben
\Gamma_{i,\min} = \{ \varphi(A )\mid A \cup \{ i\} \text{ or } A \cup \{ i^\ast\} \in \mc{C}\},\label{eq:indAcc}
\een
where $\mc{C}$ is the collection of circuits of $\mc{L}$. 
We say a Lagrangian matroid is secret sharing if the access structure induced by it for 
any $i\in [n]$ is a quantum access structure. 
(Such an access structure is monotonic and satisfies the no-cloning theorem. In terms of 
minimal access structures, it means that any two authorized sets are not disjoint.)

It is possible that a Lagrangian matroid
can induce a quantum access structure for some $i\in [n]$ but not all $i$. For simplicity we
consider the case when it induces on all $i\in [n]$.

We do not yet have a condition for which Lagrangian matroids induce quantum access structures
and which do not. We provide partial answers in both directions. First we give a necessary
condition for a Lagrangian matroid to induce a quantum secret sharing scheme. Then we give
a sufficient condition for a Lagrangian matroid to induce a secret sharing scheme. 

\begin{theorem}\label{th:necessQss}
Suppose that $G$ is a  graph wihtout loops or multi-edges and whose adjacency matrix is 
given by $A$. Let ${L}$ be a Lagrangian matroid induced by $G$ such that 
${L}$ is represented by $\left[\ba{cc}I&A\ea\right]$. If $G$ has no cycles of length $\leq 4$ and no vertices of degree 1, then the access structure induced by ${L}$ is not a valid 
quantum access structure.
\end{theorem}
\begin{proof}
A Lagrangian matroid of this type corresponds to a graph state
 whose stabilizer is given by 
\be
S= \spn{K_v \mid v\in V(G)}, \mbox{ where } K_v= X_v\prod_{i\in N(v)} Z_i
\ee
and $V(G)$ is the vertex set of $G$ and $N(v)$ is the set of neighbors of $v$.
The associated Lagrangian matroid has the representation
$\left[\ba{cc}I&A\ea\right]$. Consider access structure induced by the vertex $v$. 
\be
\Gamma_{v,\min} = \{ \varphi(A )\mid A \cup \{ v\} \text{ or } A \cup \{ v^\ast\} \in \mc{C}\}.
\ee
Of interest are two
elements in $\mc{C}$ that are induced by the generators $K_u$, where $u,w\in N(v)$.
By assumption $|N(v)|>1$. Therefore there are at least two generators $u,w\in N(v)$.
The supports of generators correspond to circuits and are of the 
$\{u  \}\cup N(u)^\ast$ and $\{w \} \cup N(w)^\ast$ respectively. Consequently the sets induced by these circuits are of the form  
$\supp(K_u) \setminus \{ v \}$ and $\supp(K_w)\setminus v$. We claim that these two sets are disjoint. Suppose that they
are not, then there exists a vertex $x\neq v$ such that $x\in \supp(K_u)   \cap \supp(K_w)$. This 
implies that $G$ has a 4-cycle contrary to assumptions. Therefore these two circuits induce
disjoint authorized sets and the induced access structure cannot be a quantum access structure.
\end{proof}

\begin{lemma}
Let $L$ be a self-dual Lagrangian matroid whose collection of circuits is given by 
$\mc{C}$. Then the collection of cocircuits of ${L}$ is given by 
 $\mc{C}^\ast = \{ C^\ast \mid C\in \mc{C} \}=\mc{C} $.
\end{lemma}
\begin{proof}
Let $\mc{B}$ be the collection of bases of the matroid. Then collection of bases of the dual
matroid is given by $\mc{B}^\ast= \{B^\ast\mid B\in \mc{B} \}$.
Let $C \in \mc{C}$ be a circuit of the matroid. Since $B^\ast$ is also an element of 
$\mc{B}$, $C$ is not a subset of $B^\ast$ for any $B\in \mc{B}$. Therefore, $C^\ast$ is in 
$\mc{C}$ as well, and $\mc{C}=\mc{C}^\ast =\{C^\ast\mid C\in \mc{C} \}$, which is precisely
the collection of circuits of the dual matroid.
\end{proof}

\begin{theorem}\label{th:selfdualLmat}
Let $\mc{L}$ be a self-dual Lagrangian matroid. Then the 
access structure $\Gamma_{i,\min}$ as defined in equation~\eqref{eq:indAcc} is a valid 
quantum access structure. 
\end{theorem}
\begin{proof}
Let $A'$ and $B'$ be two authorized sets in $\Gamma_{i,\min}$. Then there exist two circuits
$A\cup \{ a\}$ and $B\cup\{b\}$ such that $A'=\varphi(A)$ and $B'=\varphi(B)$, where 
$a,b\in \{ i,i^\ast\}$.
Suppose that $a\neq b$. We observe that $B^\ast \cup \{b^\ast\}$
must be a cocircuit of $\mc{L}$. Since $\mc{L}$ is self-dual it follows that $B^\ast \cup \{b^\ast\}$
is a circuit of $\mc{L}$. Since $\varphi(B)= \varphi(B^\ast)$, we can 
instead consider $B^\ast$. Without loss of generality we can assume that $a=b=i$. 

The self-duality of $\mc{L}$ implies that $B\cup \{i\}$ is a cocircuit of $\mc{L}$. By 
\cite[Theorem~4.2.5]{borovik03} 
it follows that 
\be
|(A\cup \{i\})\cap (B\cup \{i\})| \neq 1.
\ee
But this implies that $|A\cap B|\geq 1$ for any pair of minimal authorized sets. 
This is the necessary and sufficient condition for an access structure to be a minimal
quantum access
structure. 
\end{proof}
\begin{corollary}\label{co:selfdualLmat}
A self-dual Lagrangian matroid induces a quantum secret sharing scheme.
\end{corollary}

However, self-dual Lagrangian matroids are not the only matroids which induce valid quantum 
access structures. 
Consider the Lagrangian matroid whose representation is given by the following matrix. 
\be
\left[\ba{cccccc|cccccc}0&0&0&0&0&0&1&1&1&1&1&1\\
1&1&1&1&1&1&0&0&0&0&0&0\\
0&1&0&0&1&0&0&0&1&1&0&0\\
0&0&1&0&0&1&0&0&0&1&1&0\\
0&1&0&1&0&0&0&0&0&0&1&1\\
0&0&1&0&1&0&0&1&0&0&0&1\\
\ea \right]
\ee
The circuits of this matroid are given by 
\be
\mc{C}=\left\{\ba{l}  
\{ 1,3^\ast,4,5^\ast\} , 
\{ 1,4^\ast,5,6^\ast\}, 
\{ 1,2^\ast,5^\ast,6\},\\
\{ 1,2,3^\ast,6^\ast\},
\{ 1,2^\ast,3,4^\ast \},
\{ 1^\ast,2^\ast,4,5\},\\
\{ 1^\ast,3^\ast,5,6\},
\{ 1^\ast,2,4^\ast,6\},
\{ 1^\ast,2,3,5^\ast\},\\
\{ 1^\ast,3,4,6^\ast\},
\{ 2,3^\ast,4^\ast,5\},
\{ 3,4^\ast,5^\ast,6\},\\
\{ 2,4,5^\ast,6^\ast\},
\{ 2^\ast,3,5,6^\ast\},
\{ 2^\ast,3^\ast,4,6\}
\ea
 \right\}
\ee
The access structure induced by the treating the first coordinate as the dealer is 
given by 
\be
\Gamma_{1,\min} =\left\{ 
\ba{c}
\{2,3,4 \}, \{2,3,5 \},
\{2,3,6 \}, \{2,4,5 \},\\
\{2,4,6 \}, \{2,5,6 \},
\{3,4,5 \}, \{3,4,6 \},\\
\{3,5,6 \}, \{4,5,6 \}
\ea
\right\}
\ee
This is precisely the access structure of the $((3,5))$ threshold scheme and it can
be realized using the $[[5,1,3]]$ code. 
As this matroid is not self-dual, it shows that class of matroidal quantum secret sharing 
schemes is strictly larger than the class  induced by the class of self-dual Lagrangian 
matroids.

The dual of a matroid  $M=(J,\mc{B})$ is given by $M^\ast= (J,\mc{B}^\ast)$, where $\mc{B}^\ast =\{J\setminus B \mid B\in \mc{B} \}$. A matroid is said to be identically self-dual if $M=M^\ast$.
In \cite{pre10}, it was shown how to construct quantum secret sharing schemes from indentically self-dual matroids. 
This construction is a special case of Theorem~\ref{th:selfdualLmat}. 

\begin{lemma}
Let $M$ be an identically self-dual matroid. Then there exists a self-dual Lagrangian matroid
$L$ whose collection of bases is given by $\mc{B}(L) =\{ B\cup ([n]\setminus B)^\ast \mid B\in \mc{B}(M)\}$.
Further $L$  induces the same quantum access structure as $M$.
\end{lemma}
\begin{proof}
To
see this consider a identically self-dual matroid $M$ whose collection of bases is given by
$\mc{B}_1$. The collection of the bases for the dual matroid are given by 
$\mc{B}_1^\perp = \mc{B}_1$ because $M$ is identically self-dual. By definition 
$\mc{B}_1^\perp = \{ [n]\setminus B \mid B\in \mc{B}_1\}$. Therefore, for every basis 
$B$, $[n]\setminus B $ is also in $\mc{B}$. Now consider forming a Lagrangian matroid whose
collection of bases is given by $\mc{B}= \{B\cup ([n]\setminus B)^\ast \}$. It is Lagrangian
because the cardinality of any element in $\mc{B}$ is $n$. The self-duality of the 
symplectic matroid is a consequence of the self-duality of $M$.

By Theorem~\ref{th:selfdualLmat}, the access structure induced by $L$ is a valid quantum
access structure. We want to show that this access structure is precisely the access
structure induced by the matroid $M$. 
Recall that the access structure induced by 
$M$ is given by 
\be
\Gamma_{i,\min}^{M} =\{ A\mid A\cup \{ i\}  \in \mc{C}(M)\},
\ee
where $\mc{C}(M)$ is the collection of circuits of $M$. 

By Lemma~\ref{lm:repCkts}, the circuits of 
$L$ are either in $[n]$ or $[n]^\ast$. The restriction of $L$ to 
the transversal $[n]$ gives the matroid $M$, while the restriction to $[n]^\ast$ gives the
identically self-dual matroid $M^\ast=M$. Every circuit of $L$ contained in the restriction
$[n]$ (resp. $[n]^\ast$) is a circuit of $M$ (resp. $M^\ast$). But these exhaust the circuits 
of $L$. 
Thus the access structure induced by $L$, 
as given in \eqref{eq:indAcc}, is
exactly the same access structure as $M$.
\end{proof}

\section{Conclusion and open questions}
In this paper we have established a connection between quantum codes and symplectic matroids.
This opens a new perspective on quantum codes and has potential applications for 
quantum cryptography. 
Furthermore, this correspondence raises a number of interesting questions that are worth
pursuing. We list some of them here. 
\begin{enumerate}[1)]
\item Find representations for the graphical symplectic matroids. Alternatively, find a
criterion  to test which of these matroids are representable. 
\item  Find out if the quantum codes derived from the  symplectic matroid of a simple connected graph, have good parameters. 

\item What are the necessary and sufficient conditions for Lagrangian matroids to induce
quantum access structures? Can these be stated in terms of the graph underlying the 
Lagrangian matroid? 

\item Given a secret sharing  Lagrangian matroid, what is the associated quantum code that
realizes this access structure?
\item Define a polynomial that  captures the weight enumerator of the underlying quantum code for representable symplectic matroids.
\end{enumerate} 
We hope that the results in this paper will prompt further research into the applications of matroids for
quantum information. 
{
\section*{Acknowledgment}
I would like to thank Robert Raussendorf for many helpful discussions, and his support
and encouragement throughout this project. 
This research was sponsored by grants from NSERC, CIFAR, and MITACS.
}


\def\cprime{$'$}
\end{document}